\documentclass[orivec]{llncs}
\newcommand{\NP}{\ensuremath{\sf{NP}}\xspace}
\newcommand{\NPH}{\textrm{\textup{NP-hard}}\xspace}
\newcommand{\FPT}{\ensuremath{\sf{FPT}}\xspace}
\newcommand{\longversion}[1]{}

\usepackage[usenames,svgnames]{xcolor}
\usepackage{graphicx}
\usepackage{tikz}
\usetikzlibrary{decorations.shapes}

\usepackage{wrapfig}
\usepackage{xspace}
\usepackage[ruled,vlined,linesnumbered]{algorithm2e}
\usepackage{multicol}
\usepackage{framed}
\usepackage{boxedminipage}

\usepackage{amsmath,amssymb,amsthm,mathrsfs}
\usepackage[retainorgcmds]{IEEEtrantools}
\usepackage{pdfpages}
\usepackage{lastpage}

      \usepackage{euler}
      \usepackage[small,euler-digits]{eulervm}
\usepackage{authblk}

\usepackage{cleveref}
\usepackage{hyperref}

\hypersetup{
    pdfauthor={Neeldhara Misra},     % author
    pdfnewwindow=true,      % links in new window
    colorlinks=true,       % false: boxed links; true: colored links
    linkcolor=SlateBlue,          % color of internal links
    citecolor=FireBrick,        % color of links to bibliography
    filecolor=magenta,      % color of file links
    urlcolor=Emerald           % color of external links
}

\usepackage{url}

\usepackage[backgroundcolor=Moccasin,colorinlistoftodos]{todonotes}

\makeatletter
 \providecommand\@dotsep{5}
 \def\listtodoname{}
 \def\listoftodos{\@starttoc{tdo}\listtodoname}
 \makeatother

\newcounter{nmcomment}

\newcounter{vkcomment}

\newcommand{\NPC}{\textrm{\textup{NP-complete}}\xspace}

\newcommand{\YES}{\textsc{Yes}}
\newcommand{\NO}{\textsc{No}}

\newcommand{\defparproblem}[4]{
  \vspace{3mm}
\noindent\fbox{
  \begin{minipage}{.95\textwidth}
  \begin{tabular*}{\textwidth}{@{\extracolsep{\fill}}lr} \textsc{#1}  & {\bf{Parameter:}} #3 \\ \end{tabular*}
  {\bf{Input:}} #2  \\
  {\bf{Question:}} #4
  \end{minipage}
  }
  \vspace{2mm}
}

\newcommand{\OO}{{\mathcal O}}

\newcommand{\mhe}{\textsc{Max Happy Edges}}
\newcommand{\mhv}{\textsc{Max Happy Vertices}}

\begin{document}
% start of the contributions
%
\title{The Parameterized Complexity of Happy Colorings}
\titlerunning{The Parameterized Complexity of Happy Colorings}  % abbreviated title (for running head)
\author{Neeldhara Misra\inst{1} \and I. Vinod Reddy\inst{1}}
\authorrunning{N. Misra and IV. Reddy} % abbreviated author list (for running head)

%
\institute{Indian Institute of Technology, Gandhinagar, Palaj 382355\\
\email{neeldhara.m|reddy\_vinod@iitgn.ac.in}}

\maketitle              % typeset the title of the contribution

\begin{abstract}
Consider a graph $G = (V,E)$ and a coloring $c$ of vertices with colors from $[\ell]$. A vertex $v$ is said to be happy with respect to $c$ if $c(v) = c(u)$ for all neighbors $u$ of $v$. Further, an edge $(u,v)$ is happy if $c(u) = c(v)$. Given a partial coloring $c$ of $V$, the Maximum Happy Vertex (Edge) problem asks for a total coloring of $V$ extending $c$ to all vertices of $V$ that maximises the number of happy vertices (edges). Both problems are known to be NP-hard in general even when $\ell = 3$, and is polynomially solvable when $\ell = 2$. In [IWOCA 2016] it was shown that both problems are polynomially solvable on trees, and for arbitrary $k$, it was shown that MHE is \NPH{} on planar graphs and is \FPT{} parameterized by the number of precolored vertices and branchwidth. 

We continue the study of this problem from a parameterized prespective. Our focus is on both structural and standard parameterizations. To begin with, we establish that the problems are \FPT{} when parameterized by the treewidth and the number of colors used in the precoloring, which is a potential improvement over the total number of precolored vertices. Further, we show that both the vertex and edge variants of the problem is \FPT{} when parameterized by vertex cover and distance-to-clique parameters. We also show that the problem of maximizing the number of happy edges is \FPT{} when parameterized by the standard parameter, the number of happy edges. 
We show that the maximum happy vertex (edge) problem is \NPH{} on split graphs and bipartite graphs and polynomially solvable on cographs. 
 
\end{abstract}

\section{Introduction} 

Given an undirected vertex colored graph $G$, we say that a vertex $v$ in $G$ is \textit{happy} if $v$ and all its neighbors have same color. Along similar lines, an edge is happy if both its endpoints have same color. Given a  partially colored graph $G$ with $\ell$ colors, the \mhv{} ($\ell$-\textsc{MHV}) problem is to color the remaining vertices of graph such that number of happy vertices is maximized. The \mhe{} ($\ell$-\textsc{MHE}) problem is to color the remaining vertices of graph such that number of happy edges is maximized. 

The $\ell$-\textsc{MHE} problem generalizes the \textsc{Multiway Uncut} problem which is defined as follows.  Given a graph $G$ and a terminal set $S=\{s_1,\cdots,s_k\} \subseteq V(G)$, the goal is to find a partition $\{V_1, \cdots, V_k\}$ of $V(G)$ such that the number of edges with both ends points present in same $V_i$ is maximized. The Multiway Uncut problem is a special case of $\ell$-MHE problem, where each terminal has a unique precolor. 

Both $\ell$-MHV and $\ell$-MHE are \NPH{} \cite{zhang2015algorithmic} on general graphs for $\ell \geq 3$ and both 
$2$-MHV and $2$-MHE can be solved in polynomial time. Aravind et al. \cite{aravind2016linear} showed that both 
the problems admit linear time algorithms on trees. 
Zhang et al. \cite{zhang2015algorithmic} studied both problems
from the approximation point of view and given a $max\{1/k,\Omega(d^{-3})\}$-approximation algorithm for the $k$-MHV problem,
where $d$ is the maximum degree of the graph, and a $1/2$-approximation
algorithm for the $\ell$-MHE problem.

We initiate, in this work, the study of these problems from a parameterized perspective. The problem admits several natural parameters: the number of colors ($\ell$), the number of precolored vertices (say $t$), the number of happy vertices or edges (denoted by $k$, note that these parameters reflect the quality of the solution, and might hence be regarded as standard parameters), and various structural parameters. In particular, the linear time algorithms on trees prompt us to consider the question of whether the problem is \FPT{} when parameterized by treewidth of the graph ($w$). The work in~\cite{aravind2016linear} already establishes that the problem is \FPT{} when parameterized by treewidth and the number of precolored vertices. Since $\ell$-MHV and $\ell$-MHE are both \NPH{} even when there are only three colors used by the precoloring, the problems are para-\NPH{} by this parameter. We now proceed to describe some of the results that we obtain for various combinations of parameters. 

\paragraph{Our Contributions.} 
We continue the study of this problem from a parameterized prespective. Our focus is on both structural and standard parameterizations. To begin with, we establish that the problems are FPT when parameterized by the treewidth and the number of colors used in the precoloring, which is a potential improvement over the total number of precolored vertices. This follows from a MSO formulation but we also demonstrate a dynamic programming solution on nice tree decompositions. 

Further, we show that both the vertex and edge variants of the problem is \FPT{} when parameterized by vertex cover and distance-to-clique parameters. Observe that there is no exponential dependence here on the number of colors in the precoloring. We achieve this by not guessing all possible assignments of colors on the modulators, but just a wireframe of equivalence classes based on which vertices in the modulators recieve the same colors, and it turns out that this coarser information is sufficient to determine an optimal coloring. We also show that the problem of maximizing the number of happy edges is \FPT{} when parameterized by the standard parameter, the number of happy edges. This turns out to be a problem that reduces to the case of bounded vertex cover number. 

In the context of studying the problems on special classes of graphs, we show that the problem of maximizing the number of happy vertices is polynomially solvable on cographs. Unfortunately, our polynomial time approach for the $\ell$-\textsc{MHV} problem does not extend in any straightforward way to the  $\ell$-\textsc{MHE} problem. On the other hand, both variants of the problem turn out to be \NPH{} on split graphs and bipartite graphs. We note that some of the results shown here, particularly relating to algorithms parameterized by the treewidth of the graph, were obtained independently in~\cite{Akanksha2017,aravind2017algorithms}.
\section{Preliminaries}

In this section, we introduce the notation and the terminology that we will need to describe our algorithms. Most of our notation is standard. We use $[k]$ to denote the set $\{1,2,\ldots,k\}$. We introduce here the most relevant definitions, and use standard notation pertaining to graph theory based on~\cite{CyganFKLMPPS15,Diestel:2012vm}. 

All our graphs will be simple and undirected unless mentioned otherwise. For a graph $G = (V,E)$ and a vertex $v$, we use $N(v)$ and $N[v]$ to refer to the open and closed neighborhoods of $v$, respectively. The \emph{distance} between vertices $u,v$ of $G$ is the length of a shortest path from $u$ to $v$ in $G$; if no such path exists, the distance is defined to be $\infty$. A graph $G$ is said to be \emph{connected} if there is a path in $G$ from every vertex of $G$ to every other vertex of $G$. If $U\subseteq V$ and $G\left[U\right]$ is connected, then $U$ itself is said to be connected in $G$.  For a subset $S \subseteq V$, we use the notation $G\setminus S$ to refer to the graph induced by the vertex set $V \setminus S$.

\paragraph{Special Graph Classes.} We now define some of the special graph classes considered in this paper. A graph is \emph{bipartite} if its vertex set can be partitioned into two disjoint sets
such that no two vertices in same set are adjacent. A graph is a \emph{split graph} if its vertex set can be partitioned into  a clique and an independent set. Split graphs do not contain $C_4$, $C_5$ or $2K_2$ as induced subgraphs. \emph{Cographs} are $P_4$-free graphs, that is, they do not contain any induced paths on four vertices..

\paragraph{Parameterized Complexity.} A parameterized problem denoted as $(I,k)\subseteq \Sigma^*\times \mathbb{N}$, where $\Sigma$ is fixed alphabet and $k$ is called the parameter. We say that the problem $(I,k)$ is {\it fixed parameter tractable} with respect to parameter $k$ if there exists an algorithm which solves the problem in time $f(k) |I|^{O(1)}$, where $f$ is a computable function.
%A kernel for a parameterized problem $P$ is an algorithm which transforms an instance $(I,k)$ of $P$ to an equivalent instance $(I',k')$ in polynomial time such that $k' \leq k$ and $|I'| \leq f(k)$  for some computable function $f$. 
For a detailed survey of the methods used in parameterized complexity, we refer the reader to the texts \cite{downey2013fundamentals,CyganFKLMPPS15}.

We now define the problems that we consider in this paper.

\defparproblem{\mhv{}}{A graph $G = (V,E)$, a partial coloring $p: S \rightarrow [\ell]$ for some $S \subseteq V$, and a positive integer $k$.}{$k$}{Is there a coloring $c: V \rightarrow [\ell]$ extending $p$ such that $G$ has at least $k$ happy vertices with respect to $c$?}

\defparproblem{\mhe{}}{A graph $G = (V,E)$, a partial coloring $p: S \rightarrow [\ell]$ for some $S \subseteq V$, and a positive integer $k$.}{$k$}{Is there a coloring $c: V \rightarrow [\ell]$ extending $p$ such that $G$ has at least $k$ happy edges with respect to $c$?}

\section{Structural Parameterizations}

In this section, we explore the complexity of \mhv{} and \mhe{} with respect to various structural parameterizations. A key question here is if these problems are \FPT{} when parameterized by treewidth alone. While we do not resolve this question, we make partial progress in two ways. 

First, we show that \mhv{} and \mhe{} are both \FPT{} when parameterized by the combined parameter $(\ell + w)$, where $\ell$ is the number of colors used in the precoloring and $w$ is the treewidth of the input graph. In particular, our running time is $\OO^*((2\ell)^{(w+1)})$. It was already known that the problem is \FPT{} when parameterized by $(q + w)$, where $q$ is the number of precolored vertices. It was also known that the problem admits a linear-time algorithm on trees. In general, since $\ell \leq q$, and the treewidth of a tree is one, our result unifies these results (although the running time we obtain on trees is quadratic, rather than linear).   

Secondly, we show that  \mhv{} and \mhe{} are both \FPT{} when parameterized by the size of the vertex cover. The running time of this algorithm has a polynomial dependance on $\ell$, which is why this is not subsumed by our \FPT{} algorithm when we parameterized by $(\ell + w)$. Along similar lines, we also show that \mhv{} and \mhe{} are both \FPT{} when parameterized by the size of a clique modulator. 

The rest of this section is organized as follows. In the first subsection, we demonstrate that the problem is \FPT{} parameterized by treewidth. Next, we consider the vertex cover and distance to clique parameterizations. In the last subsection, we show how \mhe{} admits a \FPT{} algorithm for the standard parameter (the number of happy edges) by reducing it to the case of bounded vertex cover number. 

\subsection{Treewidth}
\begin{theorem}
\mhv{} and \mhe{} are both \FPT{} when parameterized by $(\ell + w)$, where $\ell$ is the number of colors used in the precoloring and $w$ is the treewidth of input graph. 
\end{theorem}
\begin{proof}
We give two different proofs to show \mhv{} and \mhe{} are both \FPT{} when parameterized by $(\ell + w)$.
The first proof uses a standard dynamic programming approach on nice tree decompositions (see \cite{CyganFKLMPPS15} for definition). 
The second proof follows from an application of Courcelle's theorem~\cite{courcelle1990graph} and the fact that the \mhv{} and \mhe{} problems can be expressed by MSO formulas for fixed $\ell$. 
%Due to space restriction the proof is presented in the full version of this paper.

%The result follows from an application of Courcelle's Theorem~\cite{courcelle1990graph} and the fact that the \mhv{} problem can be expressed in MSO. 
Suppose the input is a graph $G$ precolored using a partial precoloring function $p$ that uses $k$ colors. We briefly describe the MSO formulation here. The first expression below shows that $S$ is a set of happy vertices with respect to some partition of the graph into $k$ parts $V_1, \ldots V_k$.
\begin{eqnarray*}
\mbox{Happy}(S,V_1, \cdots V_k)&:=& \forall u \in V( u \in S \land u \in V_i \Rightarrow \forall v (adj(u,v) \Rightarrow v \in V_i))
\end{eqnarray*}

Next, we use a formulation to express that a given partition  $V_1, \cdots V_k$ respects a precoloring $p: V \rightarrow [k]$. This is shown by considering all pairs of vertices in the same color class and requiring that no two vertices in the same color class have the same color. 

\begin{eqnarray*}
\mbox{Good}_p(V_1, \cdots V_k)&:=& \bigwedge_{1 \leq i \leq k} \forall x,y \in V \left( (x \in V_i \land y \in V_i) \land (\bigvee_{1 \leq i \leq k} p(x) = i) \land (\bigvee_{1 \leq i \leq k} p(y) = i) \right)
\\
&&  \Rightarrow p(x) = p(y)
\end{eqnarray*}

Finally, following expression shows that there exists a total coloring of graph with $k$ colors that respects $p$ (recall that $k$ is the number of colors used by the precoloring $p$), such that all vertices of $S$ are happy. The first part of the expression ensures that the $k$ chosen sets form a partition and the second part of the expression ensures that the vertices in $S$ are happy.
\begin{eqnarray*}
\mbox{Happy}(S,k)&:=& \exists V_1, \cdots, V_k \left( \forall u \in V \left( \left(\bigvee_{1\leq i \leq k} u \in V_i \right) \land \left( u \in V_i \Rightarrow \neg(u \in \bigvee_{i \neq j} V_j ) \right) \right) \right) \\
&& \land~\mbox{Happy}(S,V_1, \cdots V_k)\\
&& \land~\mbox{Good}_p(V_1, \cdots V_k)
\end{eqnarray*}
The $k$-MHV problem is to find the maximum cardinality set $S$ such that: 

$$Happy(S):= max _S ~\mbox{Happy}(S,k).$$

Now by applying the Courcelle's theorem \cite{courcelle1990graph} to above formula the $k$-MHV problem can be solved in $f(k,tw) n ^{O(1)}$ time.

For the sake of completeness, we briefly describe the more direct and standard dynamic programming approach on nice tree decompositions. We describe here the semantics of the table and the update algorithm, and skip the proof of correctness (which is standard and easily checked) 
%to a full version of the paper. 
	
We explain the DP for the \mhv{} problem, noting that the details for \mhe{} are analogous. 

\sloppypar Let $(G,S,\ell,p,k)$ denote an instance of \mhv{}. Let $\mathcal{T} = (T,\{B_t\}_{t \in V(T)})$ be a nice tree decomposition of width $w$. For a vertex $g \in T$, $T_g$ denote the subtree rooted at $g$, let $\mathbb{B}_g$ denote the union of all the bags corresponding to vertices in $T_g$, and let $G_t$ denote the graph induced by $\mathbb{B}_g$, that is, $G[\mathbb{B}_g]$.

The table $\mathbb{T}_t$ corresponding to a bag $B_t$ is indexed by a pair $(r,S)$, where $r: B \rightarrow [\ell]$ is a labeling of the vertices in the bag with colors, and $S \subseteq B$ is a subset of vertices in the bag. We only consider labelings that are consistent with the precoloring $p$ on the vertices in $B$.

The table entry corresponding to $(r,S)$ at the bag $B_t$ is the maximum number of happy vertices that can be obtained by a coloring of $G_t$ that is consistent with the labeling $r$ on the vertices in $B$ and makes all the vertices in $S$ happy. If there is no such coloring, then the table entry is set to $-\infty$. 

We now briefly describe the updates for the different node types. Consider the bag $B_t$, and fix a labeling $r$ and a subset $S$ of $B_t$. Let $B_q$ denote the child node, and let $B_x$ and $B_y$ denote the children in the case of a join node. Note that the base cases are straightforward. 

\begin{itemize}
	\item \textbf{Introduce Node.} Suppose $v$ is introduced at bag $B_t$. Consider its color and check if all neighbors of $v$ in the bag have the same color as $v$. If this is not the case but $v \in S$ (or vice versa) then the value of the entry is $-\infty$. 	Otherwise, let $d$ be the value of $\mathbb{T}_q$ corresponding to the restriction of $(r,S)$ on $B_t \setminus \{v\}$. The entry in $B_t$ is either $d$ or $1+d$ depending on whether $v$ is happy or not.
\item \textbf{Forget Node.} Suppose $v$ is forgotten at bag $B_t$. 	Otherwise, consider all the entries in $\mathbb{T}_q$ corresponding to indices that are compatible with $(r,S)$, and the entry in the present bag is the maximum among all of these entries.
\item \textbf{Join Node.} Let $d_1$ and $d_2$ be the value of $\mathbb{T}_x$ and $\mathbb{T}_y$ corresponding to $(r,S)$ in $B_x$ and $B_y$, respectively. Let $d$ be $|S|$. The value of the table entry is then $d_1 + d_2 - d$.
 \end{itemize}

\end{proof}

\subsection{Vertex Cover and Distance to Clique}

\begin{theorem}
\label{th-vc}
\mhv{} and \mhe{} are both \FPT{} when parameterized by the size of the vertex cover of the input graph. 
\end{theorem}

\begin{proof}
We first consider the \mhv{} problem. Let $(G,S,\ell,p,k)$ denote an instance of \mhv{}. We will use $Q$ to refer to $V \setminus S$, the set of vertices that are not precolored by $p$. Further, let $X \subseteq V$ be a vertex cover of $G$. We use $d$ to denote $|X|$. 

The algorithm begins by guessing a partition of $X$ into at most $\min \{\ell,d\}$ parts. Note that the number of such partitions is at most $d^d$. Intuitively, we are guessing the behavior of an optimal coloring $c$ when projected on the vertex cover, where our notion of behavior is given by which vertices are colored in the same way as $c$. We also guess the subset of vertices in the vertex cover that are happy with respect to $c$.

Let us formalize the notion of a behavior associated with a partition and the subset of $X$. To begin with, let $\chi = (X_1, \ldots, X_t)$ be a fixed partition of $X$, where $t \leq d$. For vertex $v \in X$, we abuse notation and use $\chi(v)$ to denote the index of the part that $v$ belongs to in the partition $\chi$. In other words, if $\chi(v) = i$, then $v \in X_i$. A pair of vertices $u,v \in X$ such that $\chi(u) = \chi(v)$ are called equivalent --- we sometimes say that $u$ is equivalent to $v$, or that $u$ and $v$ are equivalent. Finally, let $Y \subseteq X$ be a (possibly empty) subset of the vertex cover.

We say that a coloring $c$ is valid and respects $(\chi,Y)$ if:

\begin{itemize}
	\item $c$ agrees with $p$ on $S$,
	\item every vertex $v \in Y$ is happy with respect to $c$, and
	\item for all $u,v \in X$, $c(u) = c(v)$ if and only if $u$ and $v$ are equivalent.
\end{itemize}   

Our goal now is to find a valid coloring $c$ that respects $(\chi,Y)$. Let $\lambda(v)$ denote the set of colors employed by $p$ in $N[v]$, in other words, $\lambda(v) := \{ j ~|~ \exists u \in N[v], p(u) = j \}.$

It is easy to check that there exists a valid coloring $c$ that respects $(\chi,Y)$ if and only if the following conditions, which we will refer to as $(\star)$, hold:

\begin{itemize}
	\item For any vertex $v \in Y$, $|\lambda(v)| \leq 1$.
	\item For any pair of vertices $u,v$ that are equivalent and $u,v \in Y$, if $\lambda(v) \neq \emptyset$ and $\lambda(v) \neq \emptyset$, $\lambda(v) = \lambda(u)$.
	\item For any pair of vertices $u,v$ that are not equivalent and $u,v \in Y$, we have that $N(u) \cap N(v) = \emptyset$. 
	\item For any vertex $v \in Y$, every vertex $u \in N(v) \cap X$ is equivalent to $v$.

\end{itemize}

In particular, that these conditions are necessary follow from the definition of what it means for a coloring to be valid and respect $(\chi,Y)$. The fact that they are sufficient will follow from the coloring obtained by the algorithm below.

We assume, without loss of generality, that all the conditions above are satisfied: indeed, if not, we simply reject the choice of $(\chi,Y)$. Our algorithm now proceeds to construct a coloring $c$ of $V$ that respects $(\chi,Y)$. In fact, among all such colorings, we will construct one that maximizes the number of happy vertices. To begin with, we initialize $c$ to coincide with $p$ on $S$. For convenience, we will use $U$ to refer to the set of uncolored vertices with respect to $c$. Observe that at this stage, $U = Q$, and when the algorithm finishes, we will have $U = \emptyset$. We now proceed as follows.

\paragraph{Phase 1. Identifying forced colors.} Let $v \in Y$ be such that $\lambda(v) \neq \emptyset$. Set $c(u) = j$ for all $u \in N[v] \cap Q$.  Further, set $c(u) = j$ for all $u$ equivalent to $v$ in $\chi$. At the end of this phase, let $v$ be any vertex in $Y$ that either has a precolored vertex in its closed neighborhood or has a precolored vertex that is equivalent to it. At the end of this phase, $v$ is a happy vertex. Observe that $c$ is well-defined because the conditions in $(\star)$ are true. We say that $X_i$ is \textit{pending} if $X_i \subseteq U$ at the end of Phase 1. If no $X_i$ is pending, we skip directly to Phase 3.

\paragraph{Phase 2. Coloring the pending $X_i$'s.} Let $i$ be such that $X_i$ is pending, and let $j \in [\ell]$. Let $w[i,j]$ denote the number of vertices in $(V \setminus X) \cap S$ that will be happy if all vertices of $X_i$ are colored $j$. This is simply the size of the set of vertices in the independent set precolored $j$, whose neighborhoods lie entirely in $X_i$.

Consider an auxiliary weighted bipartite graph, denoted by $H = ((A,B),E)$ with edge weights given by $w: E \rightarrow [|V|]$. This graph is constructed as follows. The vertex set $A$ contains one vertex for every $X_i$ that is pending at the end of the first phase. The vertex set $B$ contains a vertex corresponding to every element in $[\ell] \setminus [\cup_{v \in X} c(v)]$, that is to say, $B$ has one vertex for every color that is not already used on vertices in $X$ in Phase 1. The weight of the edge $(a_i,b_j)$ is simply $w[i,j]$. We find a matching $M$ of maximum weight in $H$. It is easy to see that any such matching saturates $A$, since $|B| \geq |A|$, the weights are positive, and all edges are present. 

For a pending part $X_i$, we now color all vertices in $X_i$ based on the matching $M$. In particular, if $M$ matches $a_i$ to $b_j$, then we color all vertices in $X_i$ with color $j$. At the end of this phase, all vertices in $X$ have been colored --- formally, $X \cap U = \emptyset$.

\paragraph{Phase 3. Coloring the independent set.} We say that an uncolored vertex in $V \setminus X$ is \textit{definitely unhappy} with respect to $\chi$ if it has at least two neighbors which are not equivalent with respect to $\chi$. If a vertex $v$ that is definitely unhappy has a neighbor $u$ in $Y$, then assign $c(u)$ to $v$. Observe that this coloring is well-defined, since vertices in $Y$ that had different colors, and are hence not equivalent, have disjoint neighborhoods. Arbitrarily color all the other remaining ``definitely unhappy'' vertices, namely those that only have neighbors in $X \setminus Y$. 

Similarly, we say that an uncolored vertex in $V \setminus X$ is \textit{always happy} if all of its neighbors are equivalent. For a vertex $v$ that is always happy, let $j$ be such that all neighbors of $v$ lie in $X_j$. We then color $v$ with the same color that we used to color all vertices in $X_j$. This completes the description of the construction of the coloring --- note that at this point, $U = \emptyset$.

To conclude, we count $k^\prime[\chi,Y]$, the number of happy vertices with respect to the constructed coloring $c$. Let $k^*$ denote $\max(k^\prime[\chi,Y])$, where the $\max$ is taken over all $\chi,S$ for which there exists a valid coloring that respects $\chi,Y$. The algorithm returns \YES{} if $k^\star \geq k$ and \NO{} otherwise.

\textbf{Proof of Correctness.} [Sketch.] Let $c^*$ be an arbitrary but fixed coloring of $G$ that maxmizes the number of happy vertices. Let $\chi,S$ be the behavior of $c^*$ with respect to $X$, that is, let $S$ be the set of happy vertices in $X$ with respect to $c^*$, and let $\chi := (X_1,\ldots,X_t)$ be a partition of $X$ based on the colors given by $c^*$. Let $c$ be the coloring output by the algorithm when considering the behavior $(\chi,S)$. Note that the algorithm does output some coloring based on the characterizing nature of the conditions in $(\star)$. 

It is easy to see that $c^\star$ and $c$ agree on the colors given in Phase 1 of the algorithm. Further, note that $c$ and $c^\star$ agree on the number (and even the subset) of happy vertices in $X$. Also, among all uncolored vertices of $V \setminus X$, all definitely unhappy vertices in $V \setminus X$ are not happy in $c^\star$, while all the vertices that are always happy are happy in $c$. Among the precolored vertices in $V \setminus X$, it can be verified that the maximum number of vertices that can be happy with respect to any coloring that respects the behavior $(\chi,S)$ is precisely the weight of the maximum matching obtained in Phase 2. Indeed, any such coloring is a matching in this auxiliary graph, and the number of happy vertices corresponds exactly to the weight of the matching. It follows that the number of happy vertices in $c$ is at least the number of happy vertices in $c^*$.

\textbf{Running Time Analysis.} Trying all possible choices of $(\chi,S)$ requires time proportional to $O((2d)^d)$. For a fixed choice of $(\chi,S)$, all the three phases of the algorithm are straightforward to implement in polynomial time. A maximum matching can be computed in time $O(n+\sqrt{n}m)$ on bipartite graphs.

We now turn to the \mhe{} problem. Here the algorithm is considerably simpler. Partition $E$ into two sets $E_0$ and $E_1$, where $E_0$ is the set of all edges who have both their endpoints in $X$, and $E_1 := E \setminus E_0$. We again guess the behavior of an optimal coloring in terms of how it partitions $X$ into equivalence classes. For a fixed partition $\chi$, we count all the happy edges in $E_0$ (note that this number does not depend on what colors are given to the parts, but merely the fact that all vertices in a part are equivalent). 

Having fixed a partition, we force the colors of the parts that have precolored vertices. Construct an auxiliary bipartite graph as we did in Phase 2 of the algorithm for \mhv{}, with the only difference that now the weight of the edge $(a_i,b_j)$ is based on the number of \textit{edges} that are made happy when all vertices $X_i$ are colored with color $j$. This helps us determine a coloring of the vertices in $X$. Uncolored vertices in $V \setminus X$ can now be colored greedily: for an uncolored vertex $v \in V \setminus X$, let $d_j$ denote $|N(v) \cap X_j|$, for $1 \leq j \leq t$. Note that coloring $v$ with the same color as the one used on $\max(d_j)$ makes $\max(d_j)$ edges happy. The correctness of this approach follows from the fact that this is the best we can hope for from a coloring that is consistent with the behavior specified by $\chi$.
\end{proof}

\begin{theorem}
\mhv{} and \mhe{} are both \FPT{} when parameterized by the size of a clique modulator of the input graph. 
\end{theorem}

\begin{proof}

Let $(G,S,\ell,p,k)$ denote an instance of \mhv{}. We will use $Q$ to refer to $V \setminus S$, the set of vertices that are not precolored by $p$. Further, let $X \subseteq V$ such that $C=G \setminus X$ is a clique. We use $d$ to denote $|X|$.

If $\ell > d+1$ then there exists at least two vertices $u$ and $v$ in $C$ such that $p(u) \neq p(v)$, which implies no vertex of $C$ is happy. First we guess the partition $(H,U)$ of $X$ in $O(2^d)$ time, where $H$ and $U$ denotes the happy and unhappy vertices of $X$ in optimal coloring $c$.

Let $H=(H_1, \cdots, H_t)$ be the partition of $H$ such that all vertices in set $H_i$, $i \in [t]$, are 
colored with the same color by $c$. We can guess the correct partition
in $O(d^d)$ time. Note that $N(H_i) \cap N(H_j) = \emptyset$: suppose $v \in N(H_i) \cap N(H_j)$ then the color of vertex $v$ is either different
from $col(H_i)$ or $col(H_j)$ which is a contradiction to the fact that all vertices in both $H_i$ and $H_j$ are happy. 

Since $H_i$ is happy there does not exist two vertices $u$ and $v$ in the set $H_i \cup N(H_i)$ such that $p(u) \neq p(v)$.
For each $i \in [t]$, if at least one vertex is precolored in $H_i \cup N(H_i)$ then assign same color to all vertices of $H_i \cup N(H_i)$.
For the sets $H_i \cup N(H_i)$, which do not have any precolored vertices,  assign a color which is not used so far to color any $H_i$. At the end arbitrarily color remaining vertices. 
For each possible partition $H$ and $U$ of $X$ and each possible partition  $H_1, \cdots, H_t$ of $H$, count the number  
of happy vertices and the optimal coloring $c$ is the coloring which maximizes the number of happy vertices. 

If $\ell \leq d+1$ then some of the clique vertices can be happy, but this only happens when all the clique vertices are colored by same color. Since there are at most $\ell$ colors we can guess the correct coloring of the clique in $O(\ell)$ time. 
%Now we color the sets $H_1, \cdots, H_d$ with the color in their neighborhood. If no vertex in clique is happy then we use the procedure %described for the case $l >k+1$.
Now it remained to color the set $X$, which can be done using the procedure described in case of $l>d+1$.

We now turn to the \mhe{} problem.
First we give a procedure to color $X$ when all the vertices in clique $C=G \setminus X$ are precolored and no vertex in $X$ is precolored.
Let $X=(X_1, \cdots, X_t)$ be the partition of $X$ such that all vertices in $X_i$ are colored in the same way by $c$. 

Let $w[i,j]$ denote the number of edges in $G[X_i \cup C]$ that will be happy if all vertices of $X_i$ are colored $j$. 
%This is simply the size of the set of vertices in the independent set precolored $j$, whose neighborhoods lie entirely in $X_i$.
Consider an auxiliary weighted bipartite graph, denoted by $D = ((A,B),E)$ with edge weights given by $w: E \rightarrow [|E|]$. This graph is constructed as follows. The vertex set $A$ contains one vertex for every set $X_i$ and the vertex set $B$ contains one vertex 
for every color used in clique. The weight of the edge $(a_i,b_j)$ is simply $w[i,j]$. We find a matching $M$ of maximum weight in $D$. It is easy to see that any such matching saturates $A$, since $|B| \geq |A|$, the weights are positive, and all edges are present. 
%Let $(G,S,\ell,p,k)$ denote an instance of \mhe{}.
We now color all vertices in $X_i$ based on the matching $M$. In particular, if $M$ matches $a_i$ to $b_j$, then we color all vertices in $X_i$ with color $j$. At the end of this phase, all vertices in $X$ have been colored. Now we are ready to describe the general case. 
Let $C_U$ be the number of uncolored vertices in clique. If $|C_U|\leq d+1$ then $X'=X \cup C_U$ is a vertex deletion distance to clique $C'$ of size at most $2d+1$, i.e., $G \setminus (X \cup C_U)$ is a clique. Since all the vertices of the clique $G \setminus (X \cup C_U)$ are precolored, The vertices of $X'$ can be colored using the procedure described above. So without loss of generality we assume that $|C_U|>d+1$. 
 \begin{lemma}\label{lem-all}
 If $|C_U|=n_1>d+1$, then in any optimal coloring $c$ all non precolored vertices in clique has to be colored with single color.  
 \end{lemma}
\begin{proof}
Since $\ell$ colors are used in precoloring there exists a color class  of size at least $\lfloor \frac{n-n_1}{\ell}\rfloor$ in $C$.
Assigning this color to all vertices of $C_U$ maximizes the number of happy edges, since if we color a vertex of $C_U$ with different color than others, then we loose at least $d+1$ happy edges in clique $C_U$ and can make at most $d$ edges happy. 
\end{proof}
 
The case when $\ell \leq d+1$ is easy, since we can guess optimal coloring of $X$ in time $O(d^d)$ and then it 
can be easily extended to color clique. 
If $\ell >d+1$, From Lemma~\ref{lem-all} we know that all vertices of $C_U$ gets same color, we guess this color in $O(\ell)$ time.
Now we need to color $X$ such that the number of happy edges is maximized. 
This can be done by simply applying the procedure describe in first case, where all vertices of clique are precolored. 
The running time of the algorithm is $O( \ell(d^{d})(n+\sqrt{n}m)( n + m ))$.
\end{proof}

\subsection{The Standard Parameter}

We finally show that \mhe{} is, in fact, \FPT{} when parameterized by the number of happy edges. Here we use the fact that if there are enough edges both of whose endpoints are uncolored, then we have a \YES{} instance right away. If not, the number of uncolored vertices can be shown to be bounded, and since it is safe to delete edges among precolored vertices (with some bookkeeping), the problem effectively reduces to the bounded vertex cover number scenario.  

% \begin{theorem}
% \mhe{} is \FPT{} when parameterized by $k$ and \mhv{} is \FPT{} when parameterized by $k$ and $\ell$.
% \end{theorem}
% \begin{proof}
% Due to space restriction the proof is presented in the full version of this paper.
% \end{proof}

\begin{lemma}
\mhe{} is \FPT{} when parameterized by $k$.
\end{lemma}

\begin{proof}
Let $(G,S,\ell,p,k)$ be an instance of \mhe{} and let $Q$ be
the set of vertices that are not precolored by $p$. Without loss of generality, we assume that no edge in $G[S]$ is happy, since if there are some happy edges in $G[S]$ then we remove them from $G$ and reduce the value of $k$.

If the number of edges in graph $G[Q]$ is at least $k$ then the algorithm returns \YES{}, since by coloring all  vertices in $Q$ with same color we can make at least $k$ edges happy.

If the number of vertices in $Q$ is at least $2k$ then the algorithm returns \YES{}. We can make at least $k$ edges happy as follows. Let $Q_1= \{x \in Q ~|~ N(x) \cap S \neq \emptyset\}$ and $Q_2= Q \setminus Q_1$. For all $v \in Q_1$ color $v$ with one of its neighbors color.  To color $Q_2$, we repeat the following procedure untill all the vertices of $Q_2$ are colored. For all $v \in Q_2$ which have colored neighbors in $Q_1$ assign one of its neighbors color to $v$.  

Therefore without loss of generality we assume $|Q|\leq 2k$.
Since all the edges in $G[S]$ are unhappy we can simply remove them from $G$. 
It is easy to see that in the resulting graph the set $Q$ is a vertex cover of size at most $2k$. 
Therefore the problem reduced to solving \mhe{} parameterized by vertex cover number which can be solved using 
Theorem~\ref{th-vc}.
\end{proof}

\begin{lemma}
\mhv{} is \FPT{} when parameterized by $k$ and $\ell$.
\end{lemma}

\begin{proof}
 Let $(G,S,\ell,p,k)$ denote an instance of \mhv{} and let $Q$ be
the set of vertices that are not precolored by $p$. Let $S_U \subseteq S$ and $Q_U \subseteq Q$ such that every vertex in $S_U \cup Q_U$ have at least two neighbors with different colors. Let $S_H= S \setminus S_U$ and $Q_H= Q \setminus Q_U$.

If the size of any color class in $S_H$ is at least $k$ then the algorithm returns \YES{}. 
We can make at least $k$ vertices happy by coloring all neighbors of
vertices in this color class with same color. Therefore we assume that size of each color class in $S_H$ is at most $k$, which implies size of $S_H$ is at most $\ell k$. 

If the size of set $Q_H$ is at least $\ell k$  then the algorithm returns \YES{}. 
There exists a subset $A$ of vertices of $Q_H$ of size at least $k$
such that each vertex of $A$ either has neighbors in same color class or has no colored neighbors. 
By coloring all vertices of $A$ and its uncolored neighbors with neighbors color we can make all vertices in $A$ happy.
Therefore we assume that size of $Q_H$ is less than $\ell k$.

It is easy to see that all vertices in $S_U \cup Q_U$ are unhappy, therefore we can remove the edges between them. 
Then the set $S_H \cup Q_H$ forms a vertex cover of size at most $2 \ell k$. Therefore the problem reduced to solving \mhv{} parameterized by vertex cover and number of colors, which can be solved using 
Theorem~\ref{th-vc}.
%Using the FPT algorithm of vertex cover we can show that \mhv{}
%is FPT when parameterized by $\ell +k$. 
\end{proof}

%\begin{theorem}[*]
%\mhv{} is \FPT{} when parameterized by $k$ and $\ell$.
%\end{theorem}

%\begin{proof}
%Due to the space restrictions the proof is presented in Appendix. 
% Let $(G,S,\ell,p,k)$ denote an instance of \mhe{} and let $Q$ be
% the set of vertices that are not precolored by $p$. Without loss of generality, we assume that no edge in $G[S]$ is happy, since if there are some happy edges in $G[S]$ then we remove them from $G$ and reduce the value of $k$.

% If the number of edges in graph $G[Q]$ is at least $k$ then the algorithm returns \YES{}, since by coloring all the edges in $G[Q]$ with single color we can make at least $k$ edges happy.

% If the number of vertices in $G[Q]$ is at least $2k$, first color all the edges in $G[Q]$ with same color and then color each uncolored vertex
% of $Q$ with one of its neighbors color which makes at least $k$ edges happy and the algorithm returns \YES{}.

% If $|Q| \leq 2k$, Since all the edges in $G[S]$ are unhappy we can simply remove them from $G$. 
% It is easy to see that in the resulting graph the set $Q$ is a vertex cover of size at most $2k$. 
% Therefore the problem reduced to solving \mhe{} parameterized by vertex cover number which can be solved using 
% Theorem~\ref{th-vc}
%\end{proof}
\section{Special Graph Classes}

\begin{theorem}
\mhv{} and \mhe{} are both \NPC{} on the class of bipartite and split graphs. 
\end{theorem}

\begin{proof}
The reductions for \mhv{} follow by easy modifications of the reduction in~\cite{zhang2015algorithmic}. Here, therefore, we only state the proofs for \mhe{}. We first consider the case of bipartite graphs. We reduce from \mhe{} on general graphs. 

We let $(G, \ell, S, p, k)$ be an instance of \mhe{}. Construct a bipartite graph $(H = (A,B), E)$ as follows. For every vertex $v \in V(G)$, we introduce a vertex $a_v \in A$. For every edge $e \in E(G)$, we introduce a vertex $b_e \in B$, and if $e = (u,v)$, then $b_e$ is adjacent to $a_u$ and $a_v$. The precoloring function $q$ mimics $p$ on $A$, that is, for every $u \in S$, $q(a_u) = p(u)$. We use $X$ to denote $\{a_u ~|~ u \in S\} \subseteq A$. Let $k^\prime = m + k$. Thus our reduced instance is $(H, \ell, X, q, k^\prime)$. 

%Observe that $m$ edges can always be made happy for any coloring of $A$. However, for a vertex $b_e$, both edges incident on $b_e$ are happy when the...

We now argue the equivalence. First, consider the forward direction. If $c$ is a total coloring of $V$ that makes $k$ edges happy, then we define a coloring $c^\prime$ for $H$ as follows: color $c^\prime(a_v) := c(v)$ for all $a_v \in A$. For every edge $e = (u,v) \in E$, color $b_e \in B$ according to $c(u)$. Note that for all edges $e$ in $G$ that are happy with respect $c$, two edges (namely $(b_e,a_v)$ and $(b_e,a_u)$) are happy with respect to $c^\prime$. Corresponding to all unhappy edges, $H$ has one happy edge with respect to $c^\prime$. Therefore, the total number of happy edges in $c^\prime$ is $2k + (m-k) = m+k$. 

In the reverse direction, let $c^\prime$ be a coloring of $H$ that makes at least $m+k$ edges happy. Now consider the coloring $c$ obtained as follows: $c(u) = c^\prime(a_u)$. We argue that at least $k$ edges are happy in $G$ with respect to $c$. Indeed, suppose not. Without loss of generality, assume that only $(k-1)$ edges are happy with respect to $c$. Then in $H$, there are at most $(k-1)$ vertices in $B$ that can have two happy edges incident on them, and therefore the total number of happy edges is at most $2(k-1) + (m-k+1) = m + k - 1$, which contradicts our assumption about the total number of happy edges in $H$ with respect to $c^\prime$.

We now turn to the case of split graphs. The construction is similar to the case of bipartite graphs. Construct a split graph $(H = (A,B), E)$ as follows. Let $(G, \ell, S, p, k)$ be an instance of \mhe{}. Let $T := {m \choose 2} + 1$. For every vertex $v \in V(G)$, we introduce $T$ copies of the vertex $a_v \in A$, denoted by $a_v[1], \ldots, a_v[T]$. For every edge $e \in E(G)$, we introduce a vertex $b_e \in B$, and if $e = (u,v)$, then $b_e$ is adjacent to all copies of $a_u$ and $a_v$. Finally, we add all edges among vertices in $B$, thereby making $H[B]$ a clique.

The precoloring function $q$ mimics $p$ on $A$ across all copies, that is, for every $u \in S$, $q(a_u) = p(u)$ for all copies of $a_u$. We use $X$ to denote $\{a_u ~|~ u \in S\} \subseteq A$. Let $k^\prime = T(m + k)$. Thus our reduced instance is $(H, \ell, X, q, k^\prime)$. 

We now argue the equivalence of these instances. First, consider the forward direction. If $c$ is a total coloring of $V$ that makes $k$ edges happy, then we define a coloring $c^\prime$ for $H$ as follows: color $c^\prime(a_v) := c(v)$ for all copies of $a_v \in A$. For every edge $e = (u,v) \in E$, color $b_e \in B$ according to $c(u)$. Note that for all edges $e$ in $G$ that are happy with respect $c$, $2T$ edges (namely $(b_e,a_v)$ and $(b_e,a_u)$ across all copies) are happy with respect to $c^\prime$. Corresponding to all unhappy edges, $H$ has $T$ happy edges with respect to $c^\prime$. Therefore, the total number of happy edges in $c^\prime$ is at least $2Tk + T \cdot (m-k) = T \cdot (m+k)$. 

In the reverse direction, , let $c^\prime$ be a coloring of $H$ that makes at least $T(m+k)$ edges happy. We argue that there must be at least one copy $A_i$ of $\{a_v ~|~ v \in V\}$ for which the number of happy edges with one endpoint in $B$ and one in $A_i$ is at least $(m+k)$. Indeed, suppose not. Then consider the following partition of the edges in $H$: let $E_0$ be all edges with both endpoints in $B$, and let $E_i$ be all edges with one endpoint in $B$ and the other endpoint in the $i^{th}$ copy of the vertices $\{a_v ~|~ v \in V\}$. For the sake of contradiction, we have assumed that the number of happy edges in $E_i$ is less than $(m+k)$ for all $i \in [T]$. Note that the total number of edges in $B$ is ${m \choose 2}$. Therefore, the the number of edges happy with respect to $c^\prime$ is at most:
$${m \choose 2} + T \cdot (m+k-1) = T \cdot (m+k) + {m \choose 2} - T < T(m+k),$$
where the last step follows by substituting for $T = {m \choose 2} + 1$. This leads to the desired contradiction. 

Having identified one set $E_i$ that has at least $(m+k)$ happy edges, the argument for recovering a coloring $c$ for $G$ that makes at least $k$ edges happy is identical to the case of bipartite graphs.
\end{proof}

% \begin{theorem}
% \mhv{} is polynomial time solvable on the class of cographs. 
% \end{theorem}
% \begin{proof}
% We use modular decomposition \cite{habib2010survey} technique to solve \mhv{} problem on cographs. 
% Due to space restriction the proof is presented in the full version of this paper.
% \end{proof}

We now turn to the \mhv{} problem on the class of cographs. 
We use the notion of modular decomposition to solve this problem on cographs. 

A set $M \subseteq V(G)$ is called {\it module} of $G$ if all vertices of $M$ have the same set 
of neighbors in $V(G)\setminus M$.  
The \emph{trivial modules} are $V(G)$, and $\{v\}$ for all $v$.
A prime  graph is a graph in which all
modules are trivial.
The modular decomposition of a graph is one of the decomposition techniques which was introduced by Gallai \cite{gallai1967transitiv}. The {\it modular decomposition} of a graph $G$ is a rooted tree $M_G$ 
that has the following properties:
\begin{enumerate}
\setlength{\itemsep}{1pt}
\setlength{\parskip}{0pt}
 \item The leaves of $M_G$ are the vertices of $G$.
 \item For an internal node $h$ of $M_G$, let $M(h)$ be the set of vertices of $G$ that are leaves of the subtree of 
 $M_G$ rooted at $h$. ($M(h)$ forms a module in $G$).
 \item For each internal node $h$ of $M_G$ there is a graph $G_h$ (\emph{representative graph}) with 
 $V(G_h)=\{h_1,h_2,\cdots,h_r\}$, where $h_1,h_2,\cdots,h_r$ are the 
 children of $h$ in $M_G$ and for  $1 \leq i<j \leq r$, $h_i$ and $h_j$ are adjacent in $G_h$ iff there are vertices $u \in M(h_i)$ and $v \in M(h_j)$ that are adjacent in $G$.
 \item $G_h$ is either a clique, an independent set, or a prime graph and 
 $h$ is labeled \emph{Series} if $G_h$ is clique, \emph{Parallel} if $G_h$ is an independent set, and \emph{Prime} otherwise.
\end{enumerate}
James et al. \cite{james1972graph} gave first polynomial time algorithm for 
 modular decomposition which runs in $O(n^4)$ time. 
 %Linear time algorithms to find modular decompositions are proposed 
%in \cite{cournier1994new,tedder2008simpler}.

\begin{theorem}
\mhv{} is polynomial time solvable on the class of cographs. 
\end{theorem}
\begin{proof}

The modular decomposition tree of cographs has only parallel and series nodes. 
Let $G$ be a cograph whose modular decomposition tree is $M_G$. Without loss of generality we assume that 
the root $r$ of tree $M_G$ is a series node, otherwise $G$ is not connected and the number of happy vertices in $G$ is equals to the sum 
of the happy vertices in each connected component. 
Let the children of $r$ be $x$ and $y$. Further, let the cographs corresponding to the subtrees at 
$x$ and $y$ be $G_x$ and $G_y$.

We assume that $\ell \geq 3$, otherwise we use the polynomial time algorithm of $2$-MHV problem on general graphs to find the number of maximum happy vertices. We assume that both $G_x$ and $G_y$ contains at least two vertices. Suppose,
if $G_x$ has only vertex $v$, then $v$ is the universal (adjacent to all vertices) vertex in $G$. It is easy to see that in any optimal coloring $c$ of $G$, if a vertex $u$ is happy then $c(u)=c(v)$ i.e, in any optimal coloring all happy vertices are colored with color of $v$. We can guess the color of $v$ in $O(\ell)$ time. 

We have the following three cases based on the number of colors present in $G_x$ and $G_y$ in the partial coloring $c$ of $G$.
 
%We assume that the number of precolors used in either $G_x$ or $G_y$ is at most one. 

\textbf{Number of colors in $G_x$ is zero.} Let $\ell \geq 3$ be the number of colors used in $G_y$ by $c$. 
It is easy to see that no vertex of $G_x$ is happy as each vertex of $G_x$ is adjacent to at least two vertices of  different colors. Moreover all the vertices in $G_x$ need to be colored with a single color otherwise the number of happy vertices becomes zero. We try over all possible $O(\ell)$ many ways of coloring $G_x$. In the end color all uncolored vertices in $G_y$ with the color used in $G_x$.

\textbf{Number of colors in $G_x$ is one.} 
Since $\ell \geq 3$, the number of colors (distinct from the color used in $G_x$) used in $G_y$ by $c$ is at least two. The optimal coloring is to color all the uncolored vertices of $G$ with color used in $G_x$.  

\textbf {Number of colors in $G_x$ is at least two.}
We assume that the number of colors in $G_y$ is at least two, otherwise we can use one of the above two cases by interchanging $G_x$ and $G_y$.
In this case no vertex in $G$ is happy, because every vertex of $G$ is adjacent to at least two precolored vertices having different colors. 
\end{proof}

\section{Conclusion}
In this paper, we study the \mhv{} and \mhe{} problems from the parameterized perspective. We showed that
\begin{itemize}
 \item Both the problems are \FPT{} with respect to the strutural parameters (a) Vertex cover (b) Distance to clique
 \item \mhe{} is \FPT{} when parameterized by number of happy edges in solution (standard parameter) and \mhv{} is \FPT{} when parameterized by number of happy vertices in the solution and the number of colors.
 \item Both \mhv{} and \mhe{} are $\NP$-hard on split graphs and bipartite graphs and \mhv{} is polynomially solvable on cographs 
\end{itemize}
The following are some interesting open problems.

\begin{itemize}
 \item Are \mhv{} and \mhe{} are \FPT{} when parameterized by the cluster vertex deletion number.
 \item Does \mhv{} and \mhe{} problems admit polynomial kernels when parameterized by (a) Vertex cover (b) Distance to clique.
\end{itemize}

\bibliographystyle{splncs03}
\bibliography{happyrefs.bib}

\end{document}